\documentclass[11pt,a4paper]{paper}

\usepackage{graphicx}
\usepackage{enumerate}
\usepackage{amsmath,amsthm,amsfonts}
\usepackage{microtype}
\usepackage{subcaption}
\usepackage{wrapfig}
\usepackage{fullpage}

\graphicspath{{./pictures/}}

\usepackage[usenames,dvipsnames]{xcolor}

\usepackage{xspace}

\usepackage{cite}

\newtheorem{theorem}{Theorem}[section]
\newtheorem{claim}[theorem]{Claim}

\newtheorem{lemma}[theorem]{Lemma}

\newcommand{\mathset}[1]{\ensuremath {\mathbb {#1}}}

\newcommand{\eps}{\ensuremath {\varepsilon}}

\newcommand{\R}{\mathset{R}}

\newcommand\A{{\mathcal A}}
\newcommand\B{{\mathcal B}}
\newcommand\C{{\mathcal C}}
\newcommand\D{{\mathcal D}}
\newcommand\E{{\mathcal E}}
\renewcommand\S{{\mathcal S}}

\DeclareMathOperator{\lr}{lr}
\DeclareMathOperator{\rr}{rr}

\DeclareMathOperator{\CD}{\mathcal{D}}
\DeclareMathOperator{\Q}{\mathcal{Q}}

\usepackage{ifdraft}

\title{Reachability Oracles for Directed
  Transmission Graphs\footnote{ 
  This work is supported in part by GIF
  projects 1161 and 1367, DFG project MU/3501/1 and 
  ERC StG 757609.
  A preliminary version appeared as
  Haim Kaplan, Wolfgang Mulzer, Liam Roditty, and Paul Seiferth.
  \emph{Spanners and Reachability Oracles for Directed Transmission Graphs.}
  Proc.\@ 31st SoCG, pp.~156--170.}}

\author{Haim Kaplan\thanks{School of Computer Science, Tel Aviv University,
 Israel, \texttt{haimk@post.tau.ac.il}} \and
 Wolfgang Mulzer\thanks{Institut f\"ur Informatik,
Freie Universit\"at Berlin,
  Germany,
  \texttt{[mulzer,pseiferth]@inf.fu-berlin.de}} \and
Liam Roditty\thanks{Department of Computer Science, Bar Ilan University,
  Israel
  \texttt{liamr@macs.biu.ac.il}} \and
Paul Seiferth\footnotemark[3]}
\begin{document}
\maketitle

\begin{abstract}
Let $P \subset \R^d$ be a set of $n$ points
in $d$ dimensions
such that each point $p \in P$ has an
\emph{associated radius} $r_p > 0$.
The \emph{transmission graph} $G$ for $P$ is the
directed graph with vertex set
$P$ such that there is
an edge from $p$ to $q$ if and only if
$|pq| \leq r_p$, for any $p, q \in P$.

A \emph{reachability oracle} is a data structure
that decides for any two vertices $p, q \in G$
whether $G$ has a path from $p$ to $q$.
The quality of the oracle is measured
by the space requirement $S(n)$,
the query time $Q(n)$, and the preprocessing time.
For transmission graphs of one-dimensional point sets, we 
can construct in $O(n \log n)$ time
an oracle with $Q(n) = O(1)$ and $S(n) = O(n)$.
For planar point sets, the ratio $\Psi$ between the largest and
the smallest associated radius turns out to be an important
parameter. We present three data
structures whose quality depends on $\Psi$:
the first works only for $\Psi < \sqrt{3}$ and achieves $Q(n) = O(1)$ with
$S(n) = O(n)$ and preprocessing time $O(n\log n)$;
the second data structure gives
$Q(n) = O(\Psi^3 \sqrt{n})$
 and $S(n) = O(\Psi^3 n^{3/2})$;
the third data structure is randomized with
$Q(n) = O(n^{2/3}\log^{1/3} \Psi \log^{2/3} n)$ and
$S(n) = O(n^{5/3}\log^{1/3} \Psi \log^{2/3} n)$
and answers queries correctly with high probability.
\end{abstract}

\section{Introduction}

Representing the connectivity of a graph in a space efficient, 
succinct manner, while supporting fast queries, is
one of the most fundamental data structure questions on graphs.
For an undirected graph,
it suffices to compute the connected components and
to store with each vertex a label for the respective component.
This leads to a linear-space data
structure that can  decide in constant time if any two given vertices are
connected. For directed graphs, however, connectivity is not a 
symmetric relation any more,
and the problem turns out to be much more challenging.
Thus, if $G$ is a directed graph, we say that a vertex $s$ can 
\emph{reach} a vertex $t$ if there is a directed path in $G$
from $s$ to $t$. Our goal is to construct a \emph{reachability oracle},
a space efficient data structure  that answers
\emph{reachability queries}, i.e., that determines for any pair of
query vertices $s$ and $t$ whether $s$ can reach $t$. The quality of a
reachability oracle for a graph with $n$ vertices is measured by three
parameters: the \emph{space} $S(n)$, the \emph{query time} $Q(n)$ and the
\emph{preprocessing time}.  The simplest solution stores for
each pair of vertices whether they can reach each other, leading to a
reachability oracle with $\Theta(n^2)$ space and constant query time.
For sparse graphs with $O(n)$ edges, storing just the graph and 
performing a breadth first search for a query yields an $O(n)$ space 
oracle with $O(n)$ query time.
Interestingly, other than that, we are not aware of any better solutions
for general directed graphs,
even sparse ones; see Cohen et al.~\cite{CohenHaKaZw03} for 
partial results. Thus, any result that
simultaneously achieves subquadratic space and sublinear query time would be of
great interest. A lower bound by P\v{a}tra\c{s}cu~\cite{Patrascu11} shows
that we cannot hope for $o(\log n)$ query time with $O(n)$ space in sparse 
graphs,
but it does not rule out constant time queries with slightly superlinear space.
In the absence of progress towards
non-trivial reachability oracles or better lower bounds, solutions for special
cases become important.
For directed planar graphs, after a long line of 
research~\cite{ArikatiEtAl96,Frederickson87,Djidjev96,ChenXu00,Thorup04},
Holm, Rotenberg and Thorup presented a reachability oracle with 
constant query time and $O(n)$ preprocessing time and space 
usage~\cite{Holm2015}. This data structure, as well as most other 
previous  reachability oracles,  can also return the approximate 
shortest path distance between the query vertices.

\emph{Transmission graphs} constitute a graph class that shares many
similarities with planar graphs:
let $P \subset
\R^2$ be a set of points where each point $p \in P$ has a (transmission) radius
$r_p$ associated
with it. The transmission graph has vertex set $P$ and a
\emph{directed} edge between two distinct points $p,q \in P$ if and only if
$|pq| \leq r_p$,
where $|pq|$ denotes the Euclidean distance between $p$ and $q$. 
Transmission graphs are a common model for directed sensor 
networks~\cite{KaplanEtAl15,PelegRoditty10,RickenbachEtAl09}.
In this geometric context, it is natural to consider a more general
type of query where the target point is an arbitrary point in the plane 
rather than a vertex of the graph. In this case, a vertex $s \in P$ can 
reach a \emph{point} $q \in \R^2$ if there is a \emph{vertex} $t \in P$ 
such that $s$ reaches $t$ and such that $|tq| \leq r_t$. We call such 
queries \emph{geometric reachability queries} and we call oracles that 
can answer such queries \emph{geometric reachability oracles}. To avoid
ambiguities, we sometimes use the term \emph{standard} reachability 
query/oracle when referring to the case where the query consists of 
two vertices.

\paragraph*{Our Results.}
An extended abstract of this work was presented
at the 31st International Symposium on Computational
Geometry~\cite{KaplanMuRoSe15}. 
This abstract also discusses the problem of constructing
sparse \emph{spanners} for transmission graphs. While we were preparing the
journal version, it turned out that a full description of 
our results would yield a large and unwieldy manuscript. Therefore,
we decided to split our study on transmission graphs into two parts,
the present paper that deals with the construction of efficient 
reachability oracles, and a companion paper that studies fast algorithms
for spanners in transmission graphs~\cite{KaplanEtAl15}.

In Section~\ref{sec:1d} we will see that one-dimensional 
transmission graphs
admit a rich structure that
can be exploited to construct a simple linear space geometric reachability
oracle with constant query time, and  $O(n \log n)$ preprocessing time.

In two dimensions, the situation is more involved.
Here, it turns out that the \emph{radius ratio} $\Psi$, the ratio of the
largest and the smallest transmission radius of a point in $P$, is
an important parameter.
We consider first the case where $\Psi < \sqrt{3}$. In this case,
the transmission graph has a lot
of structure: from the presence of two crossing edges $pq$ and
$rs$, we can conclude that additional edges between
$p$, $q$, $r$, and $s$ must be present.
Using this structural information, we can turn the
transmission graph into a planar graph in $O(n\log n)$ time,
while preserving the reachability relation
and keeping the number of vertices linear in $n$.
As mentioned above, for planar graphs there is a linear time
construction of a reachability oracle with linear space, and
constant query time~\cite{Holm2015}. Thus, our transformation together with
this construction yields a standard
reachability oracle with linear space, constant query time 
and $O(n\log n)$ preprocessing time. Furthermore, in the companion paper
we show that any standard reachability oracle can be
transformed into a  geometric one by paying an \emph{additive}
overhead of $O(\log n \log \Psi)$ to the query time and of $O(n \log \Psi)$
to the space~\cite{KaplanEtAl15}. We apply this transformation to the reachability oracle that we get
by planarizing the transmission graph and get a geometric oracle that requires  $O(n)$
space, $O(n\log n)$ preprocessing time, and  answers geometric queries in $O(\log n)$ time  and standard queries in $O(1)$ time.
 Section~\ref{sec:psisqrt3} presents this result.

When $\Psi \ge \sqrt{3}$, we do not know how to obtain a planar graph representing
the reachability relation of $G$.
Fortunately, we can use a theorem by Alber and Fiala that allows us
to find a small and balanced separator with respect to the area of the union
of the disks~\cite{AlberFiala04}. This leads to
a standard reachability oracle with query time $O(\Psi^3\sqrt{n})$
and space and preprocessing time $O(\Psi^3 n^{3/2})$,
see Section~\ref{sec:psiconst}.
When $\Psi$ is even larger, we
can use random sampling combined with a quadtree of logarithmic
depth
to obtain a standard reachability oracle with query time $O(n^{2/3}\log^{1/3} \Psi \log^{2/3} n)$,
space $O(n^{5/3}\log^{1/3} \Psi \log^{2/3} n)$, and
preprocessing time $O(n^{5/3}(\log \Psi + \log n) \log^{1/3} \Psi \log^{2/3} n)$. Refer to
Section~\ref{sec:psipoly}.
Again, we can transform both oracles
into  geometric reachability oracles using the result from the
companion paper~\cite{KaplanEtAl15}. Since the overhead is
additive, the transformation does not affect the performance bounds.

\section{Preliminaries and Notation}
\label{sec:prelims}

Unless stated otherwise, we let $P \subset \R^2$ denote a set of
$n$ points in the plane,
and we assume that for each point $p$, we have an
\emph{associated radius} $r_p > 0$.
Furthermore, we assume that the input is scaled so that the smallest
associated radius is $1$.
The elements in $P$ are called \emph{vertices}.
The \emph{radius ratio} $\Psi$ of $P$ is defined as
$\Psi = \max_{p \in P} r_p$ (the 
smallest radius is $1$).
Given a point $p \in \R^2$ and a radius $r$, we denote by $D(p, r)$
the closed disk with center $p$ and radius $r$. If $p \in P$, we
use $D(p)$ as a shorthand for $D(p, r_p)$. We write
$C(p, r)$ for the boundary circle of $D(p, r)$.

Our constructions for the two-dimensional reachability oracles make extensive
use of planar grids. For $i \in \{0, 1, \dots\}$,
we denote by  $\Q_i$ the \emph{grid at level $i$}. It consists of
axis-parallel squares with diameter $2^i$ that partition the
plane in grid-like fashion (the \emph{cells}).
Each grid $\Q_i$ is aligned so that the origin lies at the corner of a cell.
We assume that our model of computation allows to
 find
the grid cell containing a given point
  in constant time.

In the one-dimensional case, our construction immediately yields a geometric
reachability oracle. In the two-dimensional case, we are only able to
construct standard reachability oracles directly. However, we can use the following
result from our companion paper to transform these oracles into geometric reachability
oracles in a black-box fashion~\cite{KaplanEtAl15}.

\begin{theorem}[Theorem~4.3 in \cite{KaplanEtAl15}]
\label{thm:geometricreachability}
Let $G$ be the transmission graph for a set $P$ of $n$ points in the plane with radius
ratio $\Psi$. Given a reachability oracle for $G$ that uses $S(n)$ space and
has query time $Q(n)$, we can compute in $O(n \log n \log \Psi)$ time a
geometric reachability oracle with  $S(n) + O(n \log \Psi)$ space and
query time $O(Q(n) + \log n \log \Psi)$.
\end{theorem}

To achieve a fast preproccesing time, we need  a sparse approximation of
the transmission graph $G$. Let $\eps > 0$
be constant.
A $(1+\eps)$-\emph{spanner} for $G$ is a sparse subgraph $H \subseteq G$ such
that for any pair of vertices $p$ and $q$ in $G$ we have $d_H(p,q) \leq
(1+\eps) d_G(p,q)$ where $d_H$ and $d_G$ denote the shortest path distance in
$H$ and in $G$, respectively. In our
companion paper we show that $(1+\eps)$-spanners for transmission graphs can be
constructed efficiently~\cite{KaplanEtAl15}.
\begin{theorem}[Theorem~3.12 in~\cite{KaplanEtAl15}]
\label{thm:2dspanner}
  Let $G$ be the transmission graph for a set $P$ of
  $n$ points in the plane
  with radius ratio $\Psi$. For any fixed $\eps > 0$,
  we can compute
  a $(1+\eps)$-spanner for $G$ with $O(n)$ edges in $O(n(\log n + \log \Psi))$ time
  using  $O(n \log \Psi)$ space.
\end{theorem}

\section{Reachability Oracles for 1-dimensional Transmission Graphs}
\label{sec:1d}
In this section, we prove the existence of efficient reachability
oracles for one-dimensional transmission graphs and show that they
can be computed quickly.
\begin{theorem}
\label{thm:1doracle}
Let $G$ be the transmission graph of an $n$-point set $P \subset \R$.
Given the point set $P$ with the associated radii,
we can construct in $O(n \log n)$ time a geometric reachability oracle for
$G$ that requires  $O(n)$ space and can answer a query in  $O(1)$ time.
\end{theorem}

We begin with a simple structural observation. 
For $p \in P$, let $R_p = \{q \in P \mid p \text{ can reach } q\}$
be the set of all vertices that are reachable from $p$,
and let $I_p = \bigcup_{q \in R_p} D(q)$ denote the
union of their associated disks. Then, $I_p$ is an interval.

\begin{lemma}\label{obs:reachIB}
Let $p \in P$. There exist two points $\lr(p), \rr(p) \in \R$ such
that $I_p = [\lr(p), \rr(p)]$.
For any point $q \in \R$, the vertex $p$ can reach $q$ if and only if
$q \in [\lr(p),\rr(p)]$.
\end{lemma}

\begin{proof}
Let $\lr(p) = \min \{ s - r_s \mid s \in R_p\}$ and
$\rr(p) = \max \{ s + r_s \mid s \in R_p\}$.
From the definition, it follows that $I_p \subseteq [\lr(p), \rr(p)]$.
Conversely, let $q \in [\lr(p),\rr(p)]$, and assume
w.l.o.g that $q$ lies to the left of $p$. Let $s \in P$ be the
vertex that defines $\lr(p)$, i.e.,
$\lr(p) = s - r_{s}$.
By definition, there is a path $p = p_1 p_2 \dots p_k =  s$
from $p$ to $s$ in $G$.
Since $G$ is a transmission graph, we have $|p_i - p_{i+1}| \leq r_{p_i}$,
for $i = 1, \dots, k-1$, so the disks
$D(p_i)$ cover the entire interval $[\lr(p),p]$. Thus, there is a $p_i$
with $q \in D(p_i)$. This means that $[\lr(p), p] \subseteq I_p$.
Similarly, we have that $[p, \rr(p)] \subseteq I_p$,
so $[\lr(p), \rr(p)] \subseteq I_p$
The second statement of the lemma is now immediate.
\end{proof}

Lemma~\ref{obs:reachIB} suggests the following reachability
oracle with $O(n)$ space
and $O(1)$ query time:
for each $p \in P$, store
the endpoints $\lr(p)$ and $\rr(p)$. 
Given a query $p,q$, where $p$ is a vertex and $q$ a point in $\R$,
we return YES if and only if $q \in [\lr(p),\rr(p)]$.
It only remains to compute the interval endpoints 
$\lr(p)$ and $\rr(p)$ for all $p \in P$ efficiently.

\begin{lemma}
\label{lem:findlr}
We can find the left interval endpoint $\lr(p)$, 
for each $p \in P$, in $O(n\log n)$ total time.
An analogous statement holds for the right interval 
endpoints $\rr(p)$, for $p \in P$.
\end{lemma}

\begin{proof}
Let $p_1, p_2, \dots, p_n$ be the vertices in $P$,
sorted in ascending order of the left endpoints of
their associated disks: $p_1 - r_{p_1} \leq p_2 - r_{p_2} \leq \dots 
\leq p_n - r_{p_n}$.
Let $G'$ be the \emph{transpose graph} for $G$ in which the 
directions of all edges are reversed. We perform a
depth-first search in $G'$ with start vertex $p_1$, and 
we denote the set of all vertices
encountered during this search by $Q$. 
By construction, $Q$ contains exactly those vertices from which
$p_1$ is reachable in $G$, so
$\lr(q) = p_1$ if and only if $q \in Q$.
For each vertex $p \in P \setminus Q$, 
no vertex in $Q$ is reachable from $p$, i.e., 
$R_p \cap Q = \emptyset$.
Thus, we can repeat the procedure with the remaining
vertices to obtain all left interval endpoints.
The right interval endpoints are computed analogously.

For an efficient implementation, we store 
the $r_p$-balls around the vertices in $P$ in an 
\emph{interval tree} $T$~\cite{4M}. When a vertex $p$ is visited
for the first time, we delete the corresponding $r_p$-ball from
$T$. When we need to find an outgoing edge in $G$ from
a vertex $p$, we 
use $T$ to find one ball that contains $p$. This
can be done in $O(\log n)$ time. Since the depth-first search
algorithm traverses at most $n$ edges, this results in running
time $O(n \log n)$.
\end{proof}

\section{Reachability Oracles for 2-dimensional Transmission Graphs}
\label{sec:2d}
In the following sections we present three different geometric reachability
oracles for
transmission graphs in $\R^2$. By Theorem~\ref{thm:geometricreachability}, we
can focus on the construction of standard reachability oracles since they can
be extended easily to geometric ones. This has no effect on the space required and the time bound for a
query, expect for the oracle given
in Section~\ref{sec:psisqrt3}. This oracle  applies for $\Psi < \sqrt{3}$, it needs  $O(n
\log n)$ space and has $O(1)$ query time.
Thus, when we apply the transformation from an oracle that can answer standard reachability queries to
an oracle that can answer geometric reachability queries, we increase the query time of this oracle to $O(\log n)$.

\subsection{$\Psi$ is
less than $\sqrt{3}$}
\label{sec:psisqrt3}
Suppose that $\Psi \in [1,\sqrt{3})$. In this case, we show that
we can make $G$ planar by
first removing unnecessary edges and then resolving edge crossings
by adding $O(n)$ additional vertices.
This will not change the reachability relation between
the original vertices. The existence of efficient reachability
oracles then follows from known results for directed planar graphs.
The main goal is to prove the following lemma.

\begin{lemma}\label{lem:planarization}
Let $P$ be a set of $n$ points in $\R^2$ with $\Psi < \sqrt{3}$ and let
$G$ be the transmission graph for  $P$. We can compute, in  $O(n \log n)$ time,
 a plane
graph $H = (V, E)$ such that
\begin{enumerate}[(i)]
\item $|V| = O(n)$ and $|E| = O(n)$;
\item $P \subseteq V$; and
\item for any $p,q \in P$, $p$ can reach $q$ in $G$ if and only if $p$ can
reach $q$ in $H$.
\end{enumerate}
\end{lemma}
Given Lemma~\ref{lem:planarization}, we can obtain our reachability
oracle from known results.

\begin{theorem}
\label{thm:2doraclesmall}
Let $G$ be the transmission graph for a set $P$ of
$n$ points in $\R^2$ of radius ratio  less than
$\sqrt{3}$.
Then, we can construct in
$O(n\log n)$ time a standard reachability oracle for $G$ with $S(n) = O(n)$
and $Q(n) = O(1)$ or a geometric reachability oracle for $G$ with $S(n) =
O(n)$ and $Q(n) = O(\log n)$.
\end{theorem}
\begin{proof}
We apply Lemma~\ref{lem:planarization}
and construct the distance oracle of Holm, Rotenberg, and Thorup for
the resulting graph~\cite{Holm2015}.
This distance oracle can be constructed in linear time, it needs
linear space, and it has constant query time. The result for the geometric
reachability oracle follows from Theorem~\ref{thm:geometricreachability}.
\end{proof}

We prove
Lemma~\ref{lem:planarization} in three steps. First, we show how to
make $G$ sparse without changing the set of
reachable  pairs. Then, we show how to turn $G$ into a planar graph.
Finally, we argue that we can combine these two operations to
get the desired result.

\paragraph*{Obtaining a Sparse Graph.}
We construct a subgraph
$H \subseteq G$ with the same reachability relation as $G$ but with
$O(n)$ edges and $O(n)$ edge crossings.
The bounded number of crossings allows us to obtain a planar
graph later on.
Consider the grid $\Q_0$ (as defined in Section~\ref{sec:prelims}), 
and let $\sigma \in \Q_0$ be a grid cell.
We say that an edge  of $G$ \emph{lies in}
$\sigma$ if both endpoints are contained in
$\sigma$.
The \emph{neighborhood} $N(\sigma)$ of $\sigma$ consists of
the $7 \times 7$ block of cells in $\Q_0$ with $\sigma$ at the center.
Two grid cells are \emph{neighboring} if they lie in each
other's neighborhood.
Since a cell in $\Q_0$ has side length $\sqrt{2}/2$,
the two endpoints of every edge in $G$ must lie in neighboring
grid cells.\footnote{Since the maximum edge length in $G$ is $\sqrt{3}$,
and since $2 \frac{\sqrt{2}}{2} < \sqrt{3} < 3 \frac{\sqrt{2}}{2}$,
the neighborhood $N(\sigma)$ needs to contain three cells in each
direction around $\sigma$.}

\begin{figure}[htb]
\centering
\includegraphics{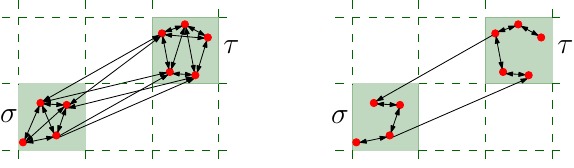}
\caption{The vertices and edges of two neighboring cells of $G$ (left) and
 of $H$ (right)}

 \label{fig:planarization}
\end{figure}

The subgraph $H$ has vertex set $P$, and we pick its edges as follows (see also
Figure~\ref{fig:planarization}):
for each non-empty cell $\sigma \in \Q_{0}$, we
set $P_\sigma = P \cap \sigma$, and
we compute the
Euclidean minimum spanning tree (EMST) $T_\sigma$ of $P_\sigma$.
For each edge $pq$ of $T_\sigma$, we add the directed edges
$pq$ and $qp$ to $H$.
Then, for every cell
$\tau \in N(\sigma)$,
we check if there are any edges from $\sigma$ to $\tau$ in $G$. If so, we
add an arbitrary such edge to $H$. We denote by
 $F$ the set of edges $pq$ such that $p$ and $q$ are in different cells.
The following lemma summarizes the properties of
$H$.

\begin{lemma}
\label{lem:pruning}
The graph $H$ has the following properties.
\begin{enumerate}[(i)]
\item for any two vertices $p$ and $q$, $p$ can reach $q$ in $G$ if and only if $p$ can reach $q$ in $H$;
\item $H$ has $O(n)$ edges;
\item $H$ can be constructed in $O(n \log n)$ time;
and
\item the straight line embedding of $H$ in the plane
has $O(n)$ edge crossings.
\end{enumerate}
\end{lemma}

\begin{proof}
(i): All edges of $H$ are also edges of $G$:
inside a non-empty cell $\sigma$, $P_\sigma$ induces
a clique in $G$, and the edges of $H$ between cells
lie in $G$ by construction. It follows that
if $p$ can reach $q$ in $H$ then $p$ can reach $q$ in $G$.

To show the converse let
$pq$ be an edge in $G$. We show that there
is a path from $p$ to $q$ in $H$.
If $pq$ lies in a cell $\sigma$ of $\Q_0$,
we take the path from $p$ to $q$ along the EMST $T_\sigma$.
If $pq$ goes from a cell $\sigma$ to
another cell $\tau$, then
there is an edge $uv$ from
$\sigma$ to $\tau$ in $H$, and
we take the path in $T_{\sigma}$ from $p$ to $u$,
then the edge $uv$, and finally the path
in $T_{\tau}$ from $v$ to $q$.

(ii): For a nonempty cell $\sigma$,
we create $|P_\sigma|-1$ edges inside $\sigma$. Furthermore,
since $|N(\sigma)|$ is constant,
there are  $O(1)$ edges between points in  $\sigma$ and points in other cells.
Thus, $H$ has $O(n)$ edges.

(iii): Since we assumed that we can find the cell for a vertex
$p \in P$ in constant time, we can easily compute the sets
$P_\sigma$, for every nonempty $\sigma \in \Q_0$,  in
$O(n )$ time.
Computing the EMST $T_\sigma$ for a cell $\sigma$ requires
$O(|P_\sigma|\log |P_\sigma|)$ time, which sums to
  $O(n \log n)$ time for all cells.
To find the edges of $F$ (i.e., edges between neighboring cells)
we build a Voronoi diagram together with a point location
structure for each set $P_\sigma$.
This takes $O(n\log n)$  time for all cells.
Let $\sigma$ and $\tau$ be two neighboring cells.
For each point in $P_\sigma$, we locate the nearest neighbor in
$P_{\tau}$ using the Voronoi diagram of $P_{\tau}$.
If there is a point $p \in P_\sigma$ whose nearest neighbor
$q \in P_{\tau}$ lies in $D(p)$, we add the edge $pq$ to $H$, and
we proceed to the next pair of neighboring cells.
Since $|N(\sigma)|$ is constant,
a point participates in $O(1)$ point location queries, each taking
 $O(\log n)$ time. The total running
time of all point location queries is $O(n \log n)$.

(iv):
Clearly each such crossing involves at least one edge of $F$ (the set of edges between points in different cells).
Each edge $e$ of $H$ intersects $O(1)$ cells $\sigma$ (this holds for edges in $F$ and trivially holds for edges inside cells). Each intersection of $e$ with an edge of $F$ must occur in one of these $O(1)$ cells that $e$ intersects.
On the other hand, each cell $\sigma$ intersects only $O(1)$ edges of $F$. So there are only $O(1)$ intersections per edge of $H$.
\end{proof}

\paragraph*{Making $G$ Planar.}
We now describe how to turn a graph $G$, embedded in the plane, into a planar graph. (This transformation can be applied to
any graph embedded in the plane. But Lemma \ref{lem:globalreachability} applies only if $G$ is a transmission graph.)
Suppose an edge $pq$
and an edge $uv$ of $G$
cross at a point $x$.
To eliminate this crossing, we add the intersection 
point $x$ as a new vertex to the graph, and we replace
$pq$ and $uv$ by the four new
edges $px$, $xq$, $ux$
and $xv$.
Furthermore, if $qp$ is an edge of $G$, we replace it by
the two edges $qx$, $xp$,
and if $vu$ is an edge of $G$, we replace it by
the two edges $vx$, $xu$. See Figure~\ref{fig:resolving}.
We say that this \emph{resolves} the crossing between $p,q,u$ and $v$.
Let $\widetilde{G}$ be the graph obtained by iteratively resolving all
crossings in $G$.
\begin{figure}[htb]
\centering
\includegraphics[scale=1.3]{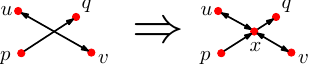}
\caption{Resolving a crossing. Since the edge $vu$ exists, we also add $vx$ and
$xu$ as edges.}
\label{fig:resolving}
\end{figure}

First, we want to show that resolving crossings keeps the \emph{local}
reachability relation between the four vertices of the crossing edges.
Intuitively speaking, the  restriction $\Psi < 3$
forces the vertices to be close together. This guarantees the existence of
additional edges between $p,q,u,v$ in $G$, and these edges
justify the new paths introduced by resolving the crossing.

To formally prove this, we first need a geometric observation.
For a point $p \in P$,
let $D(p,r)$ and $C(p, r)$ be the disk and the circle around
$p$ with radius $r$.

\begin{lemma}
\label{lem:calculus}
Let $p,q$ be two points in $\R^2$ with $|pq| = \sqrt{3}$.
\begin{enumerate}[(i)]
\item Let $a \in C(p,1) \cap C(q,1)$, and let  $b \in C(p,r) \cap C(q,r)$ for
some $r \in [1,\sqrt{3})$ such that
$a$ and $b$ lie on different
 sides of the line through $p$ and $q$.  Then $|ab| \geq r$. See Figure \ref{fig:calculusa1}.
\item Let
$\{a,b\} = C(p,\sqrt{3}) \cap C(q,1)$. Then, $|ab| > \sqrt{3}$. See Figure \ref{fig:calculusa}.
\end{enumerate}
\end{lemma}

\begin{proof}
(i): Let  $x$ be the intersection point
of the line segments $\overline{pq}$ and $\overline{ab}$.
Then $|ab| = |ax| + |xb|$.
Using that $|pa| = 1$ and $|px| = \sqrt{3}/2$, the Pythagorean Theorem
gives $|xa| = 1/2$. Similarly, we can compute $|xb|$ as a function of $r$:
 with $|pb| = r$ we get $|xb| = \sqrt{r^2 - 3/4}$.
We want to show that
\begin{align*}
 r \leq |ab|  = 1/2 + \sqrt{r^2 - 3/4}
\; \Leftrightarrow \; & r^2 \leq 1/4 + \sqrt{r^2 - 3/4} + r^2 - 3/4
\; \Leftrightarrow \;  1 \leq r^2,
\end{align*}
which holds since $r \in [1,\sqrt{3})$.

(ii): Let  $x$ be the intersection point
of $\overline{pq}$ and $\overline{ab}$.
Use the Pythagorean Theorem in the triangles $\triangle apx$ and $\triangle aqx$ in  Figure~\ref{fig:calculusa}
we get that $|ab| = 2\sqrt{11/12} > \sqrt{3}$.
\end{proof}

\begin{figure}[hbt]
\centering
\begin{subfigure}[b]{0.5\textwidth}
\centering
  \includegraphics[scale=0.53]{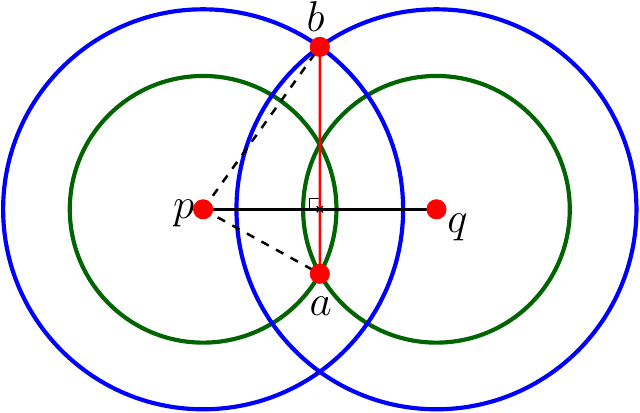}
  \caption{}
\label{fig:calculusa1}
\end{subfigure}
\begin{subfigure}[b]{0.45\textwidth}
  \includegraphics[scale=0.5]{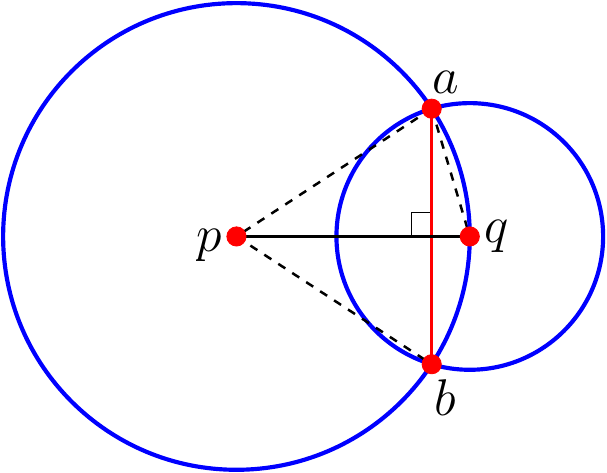}
\caption{}
\label{fig:calculusa}
\end{subfigure}
\caption{The cases (i) and (ii) of Lemma~\ref{lem:calculus}.}
\end{figure}

\begin{lemma}
\label{lem:resolving}
Suppose that $pq$ and $uv$ are
edges in a transmission graph $G$ that cross.
Let $G'\subseteq G$ be the transmission graph induced by
$p,q,u$ and $v$.  If $\Psi < \sqrt{3}$, then $p$ reaches $v$ in
$G'$ and $u$ reaches $q$ in $G'$.
\end{lemma}

\begin{proof}
We may assume that $r_p \geq r_u$. Furthermore, we assume that 
$r_q = r_v = 1$. This does
not add new edges and thus reachability in the new graph
implies reachability in $G'$.
We show that if either $u$ does not reach $q$ (case 1)
or $p$ does not reach $v$ (case 2), then $|uv| > r_u$. Hence
$uv$ cannot be an edge of $G'$ despite our assumption.

Case 1: $u$ does not reach $q$.
Then we have $p \notin D(u)$, $q \notin D(u)$, $p \notin D(v)$
and $q \notin D(v)$.
Equivalently this gives $u \notin D(p,r_u) \cup D(q,r_u)$
and $ v \notin D(p,1) \cup D(q,1)$.
Thus, the positions of $u$ and $v$ that minimize $|uv|$
are the intersections $u \in C(p,r_u) \cap C(q,r_u)$
and $v \in C(p,1) \cap C(q,1)$ on different sides of
the line through $p$ and $q$.
To further  minimize $|uv|$, observe that $|uv|$ depends on the distance
of $p$ and $q$ and that $|uv|$ strictly decreases as $|pq|$ grows, i.e.,
as $|pq|$ approaches $\sqrt{3}$.
For the limit case $|pq| = \sqrt{3}$, we are in the situation
of Lemma~\ref{lem:calculus}(i) with $a=u$ and $b=v$
and thus we would get $|uv| \geq r_u$.
But since $\Psi < \sqrt{3}$, we must have $|pq| < \sqrt{3}$ and by
strict monotonicity, it follows that $|uv| > r_u$,
as desired.

Case 2: $p$ does not reach $v$. Then we have
$u \notin D(p)$, $v \notin D(p)$, $u \notin D(q)$ and $v \notin D(q)$.
We scale everything, such that $r_p = \sqrt{3}$, and we reduce
$r_v$, $r_q$ once again to $1$.
Now, the positions of $u$ and $v$ minimizing $|uv|$ are
$\{u,v\} = C(p,\sqrt{3}) \cap C(q,1)$. As above,
further minimizing $|uv|$ gives $|pq| = \sqrt{3}$.
By Lemma~\ref{lem:calculus}(ii), we have $|uv| > \sqrt{3}$ and
thus $uv$ cannot be an edge of $G'$ (note that even after scaling we have
$r_u \leq \sqrt{3}$ since we assumed that $r_p \ge r_u$).
\end{proof}

Recall that we iteratively resolve crossings in $G$ and call the resulting graph
$\widetilde{G}$. Next, we show that for any $p, q \in P$, if $p$ can
reach $q$ in $\widetilde{G}$, then $p$ can also reach $q$ in $G$.
This seems to be a bit more difficult than what one might expect, because
when resolving the crossings, we introduce new vertices and edges
to which Lemma~\ref{lem:resolving} is not directly applicable 
(since the intermediate graph is not a transmission graph). Thus, 
a priori, we cannot exclude the possiblity that there are new 
reachabilities in $\widetilde{G}$ that use the additional 
vertices and edges.

\begin{lemma}
\label{lem:globalreachability}
Let $G$ be a transmission graph of a set $P$ of points with $\Psi < \sqrt{3}$. Let $\widetilde{G}$
be the planar graph obtained from $G$ by resolving all crossings as described above. Then,
for any two points  $p, q \in P$, $p$ can reach $q$ in $\widetilde{G}$ if and only if $p$
can reach $q$ in $G$.
\end{lemma}

\begin{proof}
If $p$ and can reach $q$ in $G$ then it immediately follows from our construction that $p$ can reach
$q$ in $\widetilde{G}$. We now prove the converse.

Each edge $e$ of $\widetilde{G}$ lies on an edge $e'$ of $G$ with the
same direction as $e$. We call $e'$ the \emph{supporting edge} of $e$.
Consider a path $\pi$ from $p$ to $q$ in $\widetilde{G}$.
A \emph{supporting switch} on $\pi$ is a pair of consecutive edges 
$\langle e,e' \rangle$ on $\pi$
such that the supporting edge of $e$ and the supporting edge of $e'$ are different.

A pair $p, q \in P$ such that $p$ can reach
$q$ in $\widetilde{G}$, but not in $G$ is called a
\emph{bad pair}.
The proof is by contradition. We assume that there exists a bad pair and  among all bad pairs,
we pick a pair $p, q$
and a path $\pi$ from $p$ to $q$ (in $\widetilde{G}$) such that
$\pi$ consists of  a
minimum number of \emph{supporting switches}, among all paths (in $\widetilde{G}$) between bad pairs.
Let $\langle e_1,e'_1 \rangle, \langle e_2,e'_2 \rangle,\dots,
\langle e_{k-1},e'_{k-1} \rangle$ be the supporting switches 
along $\pi$ and let
 $p_1q_1,\dots, p_kq_k$ be the
sequence of supporting edges as they are visited along
$\pi$ ($p_1 = p$, $q_k = q$).
That is $e_1$ is on $p_1q_1$, for $i=1,\ldots,k-2$, $e'_i$ and $e_{i+1}$ are on $p_{i+1}q_{i+1}$, and
$e'_{k-1}$ is on $p_kq_k$.
Let $x_i$ be the common vertex of $e_i$ and $e'_{i}$. The vertex $x_i$ is on the segments $\overline{p_iq_i}$ and
$\overline{p_{i+1}q_{i+1}}$.

\begin{figure}[htb]
\centering
\includegraphics[scale=1.0]{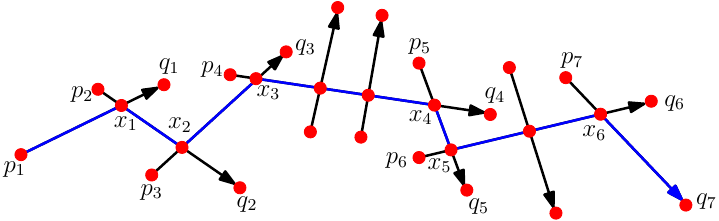}
\caption{A path (blue) with $k=7$ supporting edges that is in
$\widetilde{G}$ but not in $G$.}
\label{fig:globalreachability-path}
\end{figure}

\begin{claim}
The following holds in $G$:
\textbf{(P1)} $p_1$ reaches $q_2,\dots,q_{k-1}$;
\textbf{(P2)} $p_2,\dots,p_k$ reach $q_k$;
\textbf{(P3)}  $p_1$ and $q_1$ do not reach $p_2,\dots,p_k$; and
\textbf{(P4)} there is no edge $q_ip_i$,
for $i \geq 2$.
Furthermore, for $i = 1,\dots,k-1$, we have that
\textbf{(P5)}
the vertex $x_i$ is in the interior of
$\overline{p_iq_i}$ and
$\overline{p_{i+1}q_{i+1}}$
and \textbf{(P6)} $x_{i+1}$ lies in the interior of
$\overline{x_iq_{i+1}}$.
\end{claim}

\begin{proof}
\textbf{P1} and \textbf{P2} follow from the minimality of $\pi$,
and \textbf{P3} follows from \textbf{P2}.
For \textbf{P4}, assume
that $G$ contains an edge $q_ip_i$, for $i \geq 2$.
By \textbf{P1}, $p_1$ reaches $q_i$ in $G$ and thus $p_1$ reaches $p_i$,
despite \textbf{P3}.
For \textbf{P5}, notice that if $x_i$ is not in the interior of 
$\overline{p_iq_i}$ and
$\overline{p_{i+1}q_{i+1}}$, then
$x_i = q_i = p_{i+1}$. But then, by \textbf{P1}, $p_1$ reaches
$q_i = p_{i+1}$, despite \textbf{P3}.
\textbf{P6}
is immediate from \textbf{P5} and the fact that 
$p_{i+1}q_{i+1}$ cannot be equal to $q_{i}p_i$.
\end{proof}

By Lemma~\ref{lem:resolving}, we have $k \geq 3$,
since for two crossing edges ($k = 2$) no new 
reachabilities between the endpoints are created.
We now argue that the path $\pi$ cannot exist.
Since $p_1q_1$ and $p_2q_2$ cross,
Lemma~\ref{lem:resolving} implies that $G$ contains at least one
of
$p_1p_2,q_1p_2,p_1q_2$,
or $q_1q_2$. This is because by Lemma~\ref{lem:resolving}, in the 
induced subgraph for $p_1$, $p_2$, $q_1$, $q_2$, the vertex 
$p_1$ can reach $q_2$, and this requires that at least one of the edges
$p_1p_2,q_1p_2,p_1q_2$, or $q_1q_2$ be present.
By \textbf{P3}, neither
$p_1p_2$ nor $q_1p_2$
exist.
There are two cases, depending on whether $G$ contains
$p_1q_2$, or $q_1q_2$
(see Fig.~\ref{fig:globalreachability}).
Each case  leads to a contradiction with the minimality of $\pi$.
\begin{figure}[htb]
\centering
\includegraphics[scale=0.7]{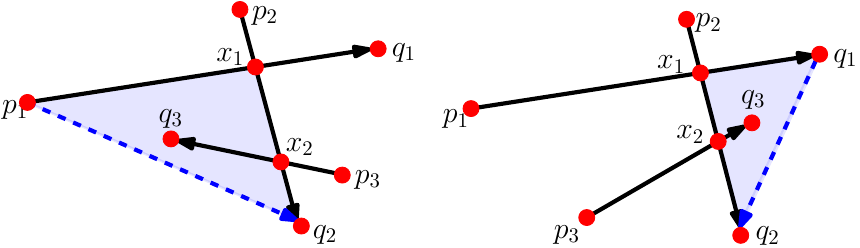}
\caption{Either $p_1q_2$ or $q_1q_2$ locks $x_3$ in
the corresponding triangle.}
\label{fig:globalreachability}
\end{figure}

\textbf{Case 1.} $G$ contains $p_1q_2$.
Consider the triangle $\triangle = \triangle p_1x_1q_2$.
Since $q_2,x_1 \in D(p_1)$, we have
$\triangle \subset D(p_1)$.
Thus, by \textbf{P3}, none of $p_2,\dots,p_k$
may lie inside $\triangle$.
By \textbf{P6}, $p_3q_3$
intersects the boundary of $\triangle$ in the line segment
$\overline{x_1q_2}$.
First, suppose that $k = 3$. In this case
$ q_3 \not\in \triangle$ (otherwise $p_1$ could reach $q_3$).
Thus, $p_3q_3$ intersects the boundary of $\triangle$
twice, so $p_3q_3$ either intersects
 $p_1q_1$ or $p_1q_2$.
In both cases, Lemma~\ref{lem:resolving} shows that $p_1$ reaches
$q_3$. Thus, we must have $k \geq 4$.

We now prove that the intersection $x_3$ of
$p_3q_3$ and $p_4q_4$ must
lie in $\triangle$.
If
 $p_3q_3$ intersects $\triangle$
once, then $q_3 \in \triangle$,
and therefore $x_3$, that by \textbf{P6} must lie on the segment
$x_2q_3$, is in $\triangle$. So assume that
$p_3q_3$ intersects
$\triangle$ twice, and let $y$ be the second intersection
point of $p_3q_3$ with the boundary of $\triangle$. We claim that $y$ follows $x_2$
along $p_3q_3$. Assume otherwise, then since by \textbf{P6}, $x_3$  follows
 $x_2$ on $p_3q_3$,
we can construct a path with fewer supporting switches than $\pi$:
If $y \in \overline{p_1x_1}$, we omit
$p_2q_2$ and
if $y \in p_1q_2$,
we omit $p_2q_2$
and substitute $p_1q_1$ by $p_1q_2$.
By the same argument, $x_3$ cannot follow $y$ on
$p_3q_3$.
Thus, $x_3$ lies on
the line segment $\overline{x_2y} \subset \triangle$.
This concludes the proof that
 $x_3 \in \triangle$.
Now, consider the segment $\overline{p_4x_3}$. Since we observed that
$p_4 \not\in\triangle$, we have that $\overline{p_4x_3}$
intersects $\triangle$, and we can again replace $\pi$ by a path with fewer supporting 
switches from $p$ to $q$.

\textbf{Case 2.} $G$ contains $q_1q_2$.
Consider the triangle
$\triangle = \triangle x_1q_1q_2$. We claim
that $\triangle \subset D(p_1) \cup D(q_1)$.
Then the  argument continues analogously  to
Case 1. In particular, 
\textbf{P3} still shows that none of $p_2,\dots,p_k$
may lie inside $\triangle$. The case $k = 3$ 
can again be ruled out, because then $p_3q_3$ would have to
intersect either $p_1q_1$ or $q_1q_2$, and 
Lemma~\ref{lem:resolving} would show that $p_1$ can reach $q_3$.
For $k \geq 4$, we can again show that 
$x_3$ would have to lie inside $\triangle$ (otherwise,
we could obtain bad pair with fewer supporting switches
by either omitting $p_2q_2$ or omitting $p_2q_2$ and 
substituting $p_1q_1$ by $q_1q_2$). Thus, by considering
the segment $\overline{p_4x_3}$, we could again find 
a bad pair with fewer supporting switches.

We now show that
that $\triangle \subset D(p_1) \cup D(q_1)$.
If $x_1 \in D(q_1)$ then $\triangle \subseteq D(q_1)$ and we are done.
Otherwise,
let $D(x_1) \subseteq D(p_1)$ be the disk with center $x_1$ and $q_1$
on its boundary. We claim that $D(x_1)$ contains $\triangle\setminus D(q_1)$.
Let $y$ be the intersection of $C(q_1)$ with $x_1q_2$.
Since $|q_1y| \ge  |q_1q_2|$, $\angle q_1yq_2 \le \pi/2$. Therefore
$\angle q_1yx_1 \ge \pi/2$ and
$|x_1y| < |x_1q_1|$.
This implies that $x_1y$ is contained in $D(x_1)$ and therefore
$\triangle\setminus D(q_1)$ is contained in $D(x_1)$ as required.
\end{proof}

\paragraph*{Putting it together.}
Let $G$ be a transmission graph of a set $P$ of points, given 
by the point set $P$ and the associated radii.
To prove Lemma~\ref{lem:planarization}, we first construct the sparse
subgraph $H$ of $G$ as in Lemma~\ref{lem:pruning} in time $O(n \log n)$.
Then we iteratively resolve the crossings in $H$ to obtain
$\widetilde{H}$.
Since $H$ has $O(n)$ crossings that can be found in $O(n)$ time,
this takes $O(n)$ time.

The graph $H$ is not necessarily a transmission graph. Therefore, we cannot 
directly apply Lemma~\ref{lem:globalreachability} to $H$
and conclude that $\widetilde{H}$ preserves the reachability relation
(between points of $P$) of $H$ and therefore of
$G$. Nevertheless, in the following lemma, we will
prove  that $\widetilde{H}$ and $G$ do have the same reachability 
relation between points of $P$.
\begin{lemma}
Let $G$ be a transmission graph on a set $P$ of points. Let $H$ be a sparse subgraph of $G$ constructed  as in Lemma~\ref{lem:pruning}
and let $\widetilde{H}$ be the planar graph obtained by resolving the crossings in $H$ as described above.
Then for any two points  $p, q \in P$, $p$ can reach $q$ in $\widetilde{H}$ if and only if $p$
can reach $q$ in $G$.
\end{lemma}
\begin{proof}
Let $\widetilde{G}$ be the graph obtained by resolving the crossings
in $G$, as described above.  If $p$ can reach $q$ in $G$, then by 
Lemma~\ref{lem:pruning}, $p$ can reach $q$ in $H$, and by the definition of
the way we resolve  crossings, $p$ can reach $q$ also in $\widetilde{H}$.

Conversely,
if $p$ can reach $q$
in $\widetilde{H}$,
then $p$ can reach $q$ in $\widetilde{G}$, because a subdivision of every edge
of $\widetilde{H}$ is contained in $\widetilde{G}$. Therefore,
by Lemma \ref{lem:globalreachability}, $p$ can reach $q$ in $G$.
\end{proof}

\subsection{Polynomial Dependence on $\Psi$}
\label{sec:psiconst}
We now present a standard reachability oracle whose performance parameters
depend polynomially on the radius ratio $\Psi$.
Together with Theorem~\ref{thm:geometricreachability} we will obtain the
following result:
\begin{theorem}
\label{thm:2doraclefixed}
Let $G$ be the transmission graph for a set $P \subset \R^2$
of $n$ points.
We can construct a geometric reachability oracle for $G$
with $S(n) = O(\Psi^3 n^{3/2})$ and $Q(n) = O(\Psi^3\sqrt{n})$ in
time $O(\Psi^3 n^{3/2})$.
\end{theorem}

Our approach is based on a geometric separator theorem for
planar disks.
Let $\CD$ be the set of disks associated with the points in $P$.
For a subset $\E$ of $\CD$ we write
$\bigcup \E := \bigcup_{D \in \E}D$ and we let
$\mu(\E)$ be the area occupied by $\bigcup \E$.
Alber and Fiala show how to find a separator for $\CD$
with respect to $\mu(\cdot)$ \cite{AlberFiala04}.

\begin{theorem}[Theorem~4.12 in~\cite{AlberFiala04}]
\label{thm:diskseparator}
There exist positive constants $\alpha < 1$ and $\beta$ such
that the following holds:
let $\CD$ be a set of $n$ disks and let  $\Psi$ be the ratio of the largest
and the smallest radius in $\CD$.
Then we can find in  $O(\Psi^2n)$ time
a partition $\A \cup \B \cup \S$ of $\CD$ satisfying
(i) $\bigcup \A \cap \bigcup \B = \emptyset$,
(ii) $\mu(\S) \leq \beta\Psi^2 \sqrt{\mu(\CD)}$ and
(iii) $\mu(\A),\mu(\B) \leq \alpha \mu(\CD)$.
\end{theorem}
Since any directed path in $G$ lies completely in $\bigcup \D$, any path from
a vertex of a disk in $\A$ to a vertex of a disk in $\B$ needs to use at least one vertex of a disk in $\S$, see
Figure~\ref{fig:separator}.
(Notice that there may not be a path from a center $p$ of a disk in $\A$ to another center $q$ of a disk in  $\A$ containing only centers of disks in $\A$. It may be that every path from $p$ to $q$ goes through a center corresponding to a disk in $\S$.)
 Since $\mu(\S)$ is small, we can approximate $\bigcup \S$
with a few grid cells. We choose the diameter of the cells small enough such that
all vertices in one cell form a clique and are equivalent in terms of reachability.
We can thus pick one vertex per cell and store the
reachability information for it. Applying this idea
recursively gives a separator tree
that allows us to answer reachability queries efficiently. The details follow.

\begin{figure}[htb]
\centering
\includegraphics[scale=0.9]{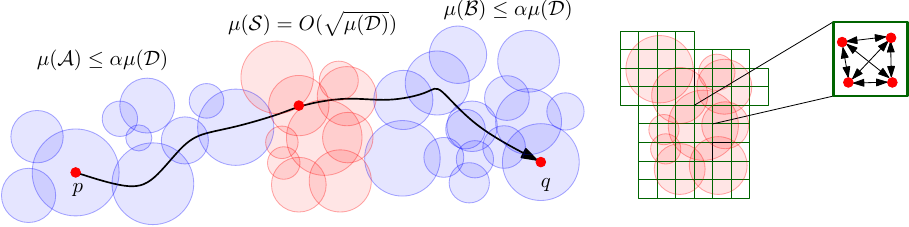}
\caption{Any path from $\A$ to $\B$ needs to use at least
one vertex of $\S$. Since $\mu(\S)$ is small, we can approximate $\bigcup \S$ with
few grid cells.}
\label{fig:separator}
\end{figure}

\paragraph*{Preprocessing Algorithm and Space Requirement.}
For the preprocessing phase, consider the grid $\Q = \Q_0$ whose
cells have diameter $1$.
All vertices in a single cell form a  clique in $G$, so the reachability information of all vertices in a grid cell is the same and it
suffices to compute this information only
for one such vertex. For each non-empty cell
$\sigma \in \Q$, we pick an arbitrary vertex $p_\sigma \in P \cap \sigma$
as the \textit{representative} of $\sigma$.
For a subset $\C\subset \D$ of disks we denote the
 set of  representatives of the non-empty cells containing centers of the disks in $\C$ by $R_\C$.

We recursively create a separator tree $T$
that contains all the required reachability information.
Each node $v$  of $T$ corresponds to an induced subgraph of the transmission graph and
the root corresponds to the entire transmission graph.
We construct the tree top down. Let $G_v$ be the subgraph associated with a node $v$ and let
$\D_v$ be the set of disks of the vertices of $G_v$.
We compute a separator $\S_v$  and subsets $\A_v,\B_v$, satisfying the conditions of Theorem~\ref{thm:diskseparator} for $G_v$.
Let $Q_v$ be all cells in $\Q$ containing centers of disks of
 $\S_v$. Let $R_v$ be the set of representatives of $Q_v$,
and let $\C_v\subset \D_v$ be all disks with centers in $Q_v$ (Note that $\C_v$ contains $\S_v$).
For each $r \in R_v$, we store all the disk centers of
$\D_v$ that $r$ can
reach and all the disk centers of
$\D_v$ that can reach $r$ in $G_v$.
We recursively compute separator trees for  the transmission graphs induce by the centers of $\A_v \setminus \C_v$ and
the centers of
$\B_v \setminus \C_v$. The roots of these trees are children of $v$ in $T$.

To obtain the  required reachability information at a node $v$ of $T$,
we compute a $2$-spanner $H_v$ for the transmission
graph $G_v$, as in Theorem~\ref{thm:2dspanner}.
Since we are only
interested in the reachability properties of the spanner,
$\eps = 1$ (or any constant) suffices.
For each $r \in R_v$,
we compute a BFS tree in $H_v$ with root $r$. Next, we reverse all
edges in $H_v$, and we again compute BFS-trees for all $r \in R_v$
in the transposed graph. This gives the required reachability
information for $v$.

As $T$ has $O(\log n)$ levels,
the total running time for computing the spanners is
$O(n\log n(\log n + \log \Psi))$.
Since the spanners are sparse, the time for computing a single BFS-tree
associated with a node $v$ is $O(|\D_v|)$. It follows that
 the time for computing all
BFS-trees at $v$ is $O(|\D_v|\cdot|R_v|)$ and the time to compute all BFS trees of all nodes  of
the separator tree $T$ is
$O(\sum_{v\in T} |\D_v|\cdot|R_v|)$. To bound this sum,
we need the
following lemma.

\begin{lemma}
\label{lem:diskunioncover}
Let $\E$ be a set of $n$ disks with radius at least $1$. Then
the number of cells in $\Q_0$ that intersect
$\bigcup \E$ is $O(\mu(\E))$.
\end{lemma}

\begin{proof}
Let $S \subset \Q_0$ be the set of all cells that intersect
$\bigcup \E$. For $\sigma \in S$, the \emph{neighborhood}
of $\sigma$ is defined as the region consisting of
$\sigma$ and its eight surrounding cells.
Let $S' \subseteq S$ be a maximal subset of cells
in $S$ whose neighborhoods are pairwise disjoint.
Then, $|S| = O(|S'|)$.
Now, let $\sigma \in S'$.
Since all disks in $\E$ have radius at
least $1$, there is a disk $D'$ (not necessarily
in $\E$) of radius
exactly $1/(2\sqrt{2})$ such that $D' \subseteq \bigcup \E$
and such that $D'$ intersects the boundary of $\sigma$.
Thus, the intersection of $\bigcup \E$ and the neighborhood
of $\sigma$ contributes at least $\mu(D') = \Omega(1)$ to $\mu(\E)$.
Since the neighborhoods for the cells in $S'$ are pairwise disjoint,
it follows that $|S| = O(|S'|) = O(\mu(\E))$, as claimed.
\end{proof}

Now, by Lemma~\ref{lem:diskunioncover}, we have $|R_v| = O(\mu(\S_v))$.
Thus, if we denote by $L_i$ the nodes of the separator tree at level
$i$ of the recursion, we get that
the sum $\sum_v |\D_v|\cdot|R_v|$ is proportional to
\begin{align*}
\sum_{i\geq 0} \sum_{v \in L_i} |\D_v| \cdot \mu(\S_v)
&\leq \sum_{i \geq 0} \sum_{v \in L_i} |\D_v| \cdot \beta\Psi^2\sqrt{\mu(\D_v)} & \text{(by Theorem~\ref{thm:diskseparator}(ii))}\\
&= \sum_{i \geq 0} \sum_{v \in L_i} |\D_v| \cdot \beta\Psi^2\sqrt{\alpha^i\mu(\D)}& \text{(by Theorem~\ref{thm:diskseparator}(iii))}\\
&= \beta\Psi^2 \sqrt{\mu(\D)} \sum_{i \geq 0}  \alpha^{i/2} \sum_{v \in L_i} |\D_v| \\
&\leq \beta\Psi^2 n \sqrt{\mu(\D)} \sum_{i \geq 0}  \alpha^{i/2} &
\text{(the $\D_v$ at a level are disjoint)} \\
&=  O(\Psi^3 n^{3/2})  & (\mu(\D) = O(\Psi^2n), \alpha < 1) \\
\end{align*}
Thus, the total preprocessing time is
$O(n\log^2 n + n\log \Psi + \Psi^3 n^{3/2}) =
O(\Psi^3 n^{3/2})$.
The space requirement is also bounded by
the preprocessing time.

\paragraph*{Query Algorithm.}
Let $p,q \in P$ be given. We assume that $p$ and $q$ are the representatives of their cells. (Otherwise
we replace either $p$ or $q$ by its representative.)
Let $v$ and $w$ be the nodes in $T$ with
$p \in R_{v}$ and $q \in R_{w}$.
Let $u$ be least common ancestor of $v$ and $w$. We can find $u$ by
 walking up the tree starting from $v$ and $w$ in $O(\log n)$ time.
Let $L$ be the  path from $u$ to the root
of $T$. We check for each $r \in \bigcup_{x\in L} R_{x}$
whether $p$ can reach $r$ and whether $r$ can reach $q$. If so,
we return YES.
If there is no such vertex
$r$ then we return NO.
Since $|R_x|$ increases geometrically along $L$, the running time
is dominated by the time for processing the root, which is
$O(\Psi^2\mu(\D)^{1/2})$.
Bounding $\mu(\D)$ by $O(\Psi^2n)$,  we get that the total query time is
$O(\Psi^3\sqrt{n})$.

It remains to argue that our query algorithm is correct.
By construction, it follows
that we return YES only if there is a path from $p$ to $q$.
Now, suppose there is a path $\pi$ in $G$ from $p$ to $q$,
where $p$ and $q$ are representatives of their grid cells with $p\neq q$.
Let $v, w$ be the nodes in $T$ with
$p \in R_{v}$ and $q \in R_{w}$. Let $u$
be their least common ancestor, and $L$ be the path from $u$
to the root.
By construction, $\bigcup_{x \in L} \S_x$ contains a disk $D(r)$
of a vertex $r$ in $\pi$.
Let $x$ be the node of $L$ closest to the root such that
$\S_x$ contains such a disk, and let $r$ be a vertex on $\pi$ with
$D(r) \in \S_x$.
Let $r'$ be the representative of the cell $\sigma$ containing $r$.
Since the vertices in $\sigma$ constitute a clique,
$p$ can reach $r'$ and $r'$ can reach $q$ in
$G_x$.
Thus, when walking along $L$, the algorithm will discover $r'$ and
the path from $p$ to $q$.
Theorem~\ref{thm:2doraclefixed} now follows.

\subsection{Logarithmic Dependence on $\Psi$}
\label{sec:psipoly}

Finally, we improve the dependence on $\Psi$ to be logarithmic,
at the cost of a slight increase at the exponent of $n$.
We prove the following theorem by constructing a standard reachability oracle and
then using Theorem~\ref{thm:geometricreachability}.
\begin{theorem}
\label{thm:2doraclebounded}
Let $G$ be the transmission graph for a
set $P$ of $n$ points in the plane.
We can construct a geometric reachability oracle for $G$
with $S(n) = O(n^{5/3}\log^{1/3}\Psi \log^{2/3} n)$ and
$Q(n) = O(n^{2/3}\log^{1/3} \Psi \log^{2/3} n)$ that answers all
queries correctly with high probability.
The preprocessing time is
$O(n^{5/3}(\log \Psi + \log n)\log^{1/3} \Psi \log^{2/3}n )$.
\end{theorem}

We scale everything such that the smallest radius in
$P$ is $1$.
Our approach is as follows: let $p,q \in P$.
If there is a $p$-$q$-path  with  ``many'' vertices, we detect
this by taking a large enough
random sample $S \subseteq P$ and by storing the
reachability information for every vertex in $S$.
If there is a path from $p$ to $q$ with ``few'' vertices, then $p$
must be ``close'' to $q$, where ``closeness'' is defined relative to
the largest radius along the path. The radii of the point of $P$ can lie in
$O(\log \Psi)$ different scales, and for each scale we
store local information to find such a ``short'' path.

\paragraph*{Long Paths.}
Let $0 < \alpha < 1$ be a parameter to be
determined later.
First, we show that a random sample can be used to detect paths with
many vertices.

\begin{lemma}\label{lem:sampling}
We can sample a set $S \subset P$ of size $O(n^\alpha \log n)$ such that
the following holds with  probability
at least $1-1/n^2$: For any two points
$p,q \in P$, if there is a path $\pi$ from $p$ to $q$ in $G$ with
at least $n^{1-\alpha}$ vertices, then $\pi \cap S \neq \emptyset$.
\end{lemma}

\begin{proof}
We take $S$ to be a random subset of size
 $m = 4n^\alpha \ln n$ vertices from $P$.
Now fix $p$ and $q$ and let $\pi$ be a path from $p$ to $q$ with
$k \geq n^{1-\alpha}$ vertices.
The probability that $S$ contains no vertex from $\pi$ is
\begin{align*}
\frac{\binom{n-k}{m}}{\binom{n}{m}} & =\frac{(n-m)(n-m-1)\cdots(n-m-k+1)}{n(n-1)\cdots (n-k+1)}   \\
& = \left( 1-\frac{m}{n} \right) \left( 1-\frac{m}{n-1} \right) \cdots \left( 1-\frac{m}{n-k-1} \right) \le
(1 - m/n)^k \leq e^{-mk/n} \leq 1/n^4,
\end{align*} by our choice of $m$.
Since there are $n(n-1)$ ordered vertex pairs, the union bound
shows that the probability that $S$ fails to detect a pair of
vertices connected by a long path is at most
$n(n-1)/n^4 \leq 1/n^2$.
\end{proof}

We draw a sample $S$ as in Lemma~\ref{lem:sampling}, and for
each $s \in S$, we store two Boolean arrays that indicate for each
$p \in P$ whether $p$ can reach $s$ and whether $s$ can reach
$p$. This requires $O(n^{1+\alpha} \log n)$ space.
It remains to deal with vertices that are connected by a path with
fewer than $n^{1-\alpha}$ vertices.

\paragraph{Short Paths.}
Let $L = \lceil \log \Psi \rceil$. We consider the $L$ grids
$\Q_0,\dots,\Q_L$ (recall that the cells in $\Q_i$ have
diameter $2^i$). For each cell $\sigma \in \Q_i$, let
$R_\sigma \subseteq P$ be the vertices $p \in P \cap \sigma$ with $r_p \in
[2^i,2^{i+1})$.
The set $R_\sigma$ forms a clique in $G$, and for each
$p \in R_\sigma$, the disk $D(p)$ contains the cell
$\sigma$. For every $i = 0, \dots, L$ and for every $\sigma \in \Q_i$
with $R_\sigma \neq \emptyset$,
we fix an arbitrary  \emph{representative point}
$r_\sigma \in R_\sigma$.

The \emph{neighborhood} $N(\sigma)$ of $\sigma\in \Q_i$ is defined
as the set of all cells in $\Q_i$ that have distance at most
$2^{i+1}n^{1-\alpha}$ from $\sigma$. We have
$|N(\sigma)| = O(n^{2-2\alpha})$.
Let $P_\sigma \subseteq P$ be the vertices
that lie in the cells of $N(\sigma)$.

For every vertex $p \in P$, and
for every $i \in \{0, \dots, L\}$ we store two sorted lists of representative of cells
$\sigma \in \Q_i$ such that $p \in P_\sigma$. The first list
contains all  representatives $r_\sigma$, such that  $p \in P_\sigma$ and $p$ can reach  $r_\sigma$.
The second list contains all
 representatives $r_\sigma$,  such that  $p \in P_\sigma$ and  $r_\sigma$  can reach $p$.
A vertex $p$ belongs to $O(n^{2-2\alpha}\log \Psi)$
sets $P_\sigma$, so the total space is
$O(n^{3-2\alpha}\log \Psi)$.

\paragraph*{Performing a Query.}
Let $p, q \in P$ be given. To decide whether $p$ can reach $q$,
we first check the Boolean tables for all $O(n^{\alpha}\log n)$ points
in $S$. If there
is an $s\in S$ such that $p$ reaches
$s$ and $s$ reaches $q$, we return YES. If not,
for $i \in \{0, \dots, L\}$,
we consider the list of representatives
that are reachable from $p$ in their neighborhood at level $i$ and
the list of representatives that can reach $q$ in their
neighborhood at level $i$. We check
whether these lists contain a common element. Since the lists are
sorted, this can be done in time linear in their size.
If we find a common representative for some $i$, we return YES.
Otherwise, we return NO.

We now prove the correctness of the query algorithm. First note that we
return YES, only if there is a path from $p$ to
$q$.  Now suppose that there is a path $\pi$ from $p$ to $q$.
If $\pi$ has at least $n^{1-\alpha}$ vertices,
then by Lemma~\ref{lem:sampling}, the sample $S$ hits $\pi$
with probability at least $1-1/n^2$, and the algorithm returns YES.
If $\pi$ has less than $n^{1-\alpha}$ vertices, let $r$ be the
vertex of $\pi$ with the largest radius,
and let $i$ be such that the radius of $r$ lies in $[2^i, 2^{i+1})$.
Let $\sigma$ be the cell of $\Q_i$ that contains $r$.
Since $\pi$ has at most $n^{1-\alpha}$ vertices, and since each edge
of $\pi$ has length at most $2^{i+1}$, the path $\pi$ lies
entirely in  $P_\sigma$ and in particular
 both $p$ and $q$ are in
$P_\sigma$. Since $r \in R_\sigma$ and since
$R_\sigma$ forms a clique in $G$, the representative point
$r_\sigma$ of $\sigma$ can be reached from $p$ and can reach $q$.
It follows from the definition of the sorted lists of representatives stored with $p$ and $q$, that
 $r_\sigma$ is
contained in the list of representatives  reachable  from $p$ and in the
list of representatives that can reach $q$.
Our query algorithm detects this when it checks
whether the corresponding
lists for $p$ and $q$ at level $i$, have a nonempty intersection.

\paragraph*{Time and Space Requirements.}
We consider first the query time. To test if there is a long path from
$p$ to $q$ we traverse $S$, and for every $s\in S$ we test,
 in $O(1)$ time, whether
$p$ can reach $s$ and whether $s$ can reach $q$.
This takes $O(|S|)= O(n^\alpha \log n)$ time.
To test if there is a short path from $p$ to $q$ we use
the lists of reachable representatives associated with $p$ and $q$ at each of the $O(\log \Psi)$ grids.
At each level
we step through two lists of size $O(n^{2-2\alpha})$. So
in total we spend $O(n^{2-2\alpha}\log\Psi)$ time.
We choose $\alpha$ to balance the times we spend to detect short and long paths.
That is $\alpha$ satisfies
\[
  n^{\alpha} \log n = n^{2-2\alpha}\log\Psi \Leftrightarrow
  n^{\alpha} = n^{2/3} (\log \Psi / \log n)^{1/3}.
\]
This yields $Q(n) = O(n^{2/3}\log^{1/3} \Psi \log^{2/3} n)$.
This choice of $\alpha$ results in a space bound of
$O(n^{5/3}\log^{1/3} \Psi \log^{2/3} n)$.

For the preprocessing algorithm, we first compute the reachability arrays
for each $s \in S$. To do so, we build a 2-spanner $H$ for $G$
as in Theorem~\ref{thm:2dspanner} in  $O(n (\log n + \log \Psi))$ time.
Then,
for each $s \in S$ we perform a BFS search in $H$ and its transposed graph.
This gives all vertices that $s$ can reach and
all vertices that can reach $s$
in $O(n^{5/3}\log^{1/3}\Psi \log^{2/3} n )$ total time.
For the short paths, the preprocessing algorithm goes as follows:
For each $i = 0,\dots, L$ and
for each cell $\sigma \in \Q_i$ that has a representative $r_\sigma$,
we compute a 2-spanner $H_\sigma$ as in Theorem~\ref{thm:2dspanner} for
$P_\sigma$.
For each representative $r_\sigma$, we do a BFS search in $H_\sigma$ and
the transposed graph, each starting from $r_\sigma$. This gives all
$p \in P_\sigma$ that can reach $r_\sigma$
and that are reachable from $r_\sigma$ via a short path.
The running time is dominated by the time for constructing the
spanners. Since each point $p \in P$ is contained in
$O(n^{2-2\alpha}\log \Psi) = O(n^{2/3} \log^{1/3}\Psi \log^{2/3} n)$
different $P_\sigma$, and
since constructing $H_\sigma$ takes
$O(|P_\sigma|(\log \Psi + \log |P_\sigma|))$
time, the bound on the  preprocessing time stated in
Theorem~\ref{thm:2doraclebounded} follows.

\section{Conclusion}
Transmission graphs constitute a natural class of
directed graphs for which non-trivial reachability
oracles can be constructed. As mentioned in the
introduction, it seems to be a very challenging
open problem to obtain similar results for general directed
graphs.  We believe that our results only scratch
the surface of the possibilities offered by transmission graphs,
and several interesting open problems remain.

All our results on 2-dimensional transmission graphs depend on the radius ratio $\Psi$.
Whether this dependency can be avoided is a major open question.
Our most efficient reachability oracle is for $\Psi < \sqrt{3}$. In this case
the reachability relation in a transmission graph with $n$ vertices can be represented
by the reachability relation in a planar graph with $O(n)$ vertices.
However, it is not clear to us that the upper bound of $\sqrt{3}$ in this result is
tight.
Can we obtain a similar construction for, say, $\Psi = 100$? Is there
 a way to represent the reachability relation in \emph{any} transmission graph, regardless of $\Psi$,
by the reachability relation in a planar graph with $o(n^{2})$ vertices?
This would immediately imply a non-trivial reachability oracle for any value
of $\Psi$.

Conversely, it is interesting to see if we can represent
the reachability relation of an arbitrary directed graph
using a transmission graph. If this is possible, the relevant questions
are how many vertices such a transmission graph must have, what is the required
radius ratio, and how fast can we compute it. A representation with not too many
vertices and low radius ratio would lead to efficient reachability oracles for
general directed graphs.

\paragraph*{Acknowledgments.}
We like to thank G\"unter Rote and the anonymous reviewers 
for valuable comments, in particular for pointing out a
drastic simplification for the one-dimensional reachability oracle.

\bibliographystyle{abbrv}
\bibliography{literature}
\end{document}